\documentclass[runningheads]{llncs}

\usepackage{amsmath, bm}
\usepackage{amsfonts}
\usepackage{mathtools}
\usepackage{todonotes}
\usepackage{amssymb}
\usepackage{graphicx}
\usepackage{algorithm}
\usepackage{algpseudocode}
\usepackage{booktabs}
\usepackage{nicefrac}
\usepackage{siunitx}
\usepackage{thmtools}
\usepackage{hyperref}
\usepackage{orcidlink}

\usepackage{syntax}
\usepackage{pgfplots}
\pgfplotsset{compat=newest}

\usepackage{caption}
\usepackage{subcaption}

\usepackage{mleftright}
\usepackage{tikz}

\usepackage{listings}
\usepackage{lstautogobble}

\definecolor{keywordsColor}{RGB}{97, 0, 71}
\definecolor{commentsColor}{RGB}{54, 54, 54}

\definecolor{backcolour}{rgb}{0.95,0.95,0.92}
\lstdefinestyle{appendix-code-style}{
    basicstyle=\ttfamily\footnotesize
}

\lstdefinelanguage{OurLanguage}{
    alsoletter={:,=, <, >, &, |},
    keywords={while, end, types, if, else, else:, and, or, not},
    morekeywords={=, <, >, <=, >=, ==, !=, &&, ||},
    basicstyle={\ttfamily\small\normalfont},
    keywordstyle={\color{keywordsColor}\ttfamily\bfseries},
    comment=[l]{\#},
    commentstyle={\color{commentsColor}\ttfamily},
    autogobble=true,
    mathescape=true
}
\lstdefinestyle{program}{basicstyle=\small\ttfamily,keywordstyle=\bfseries}

\definecolor{slytherin}{RGB}{42,98,61}

\definecolor{scarlet}{rgb}{1.0, 0.13, 0.0}
\definecolor{ashgrey}{rgb}{0.7, 0.75, 0.71}

\definecolor{algo1}{HTML}{4C0033}
\definecolor{algo2}{HTML}{01024E}

\newcommand{\Polar}{{\texttt{Polar}}}

\newcommand{\prog}{\ensuremath{\mathcal{P}}}

\newcommand{\diffp}{\ensuremath{\partial_p}}

\newcommand{\vars}[1][\prog]{\text{Vars}(#1)}

\newcommand{\params}{\text{Params}(\prog)}

\makeatletter
\newcommand{\oset}[3][0ex]{%
  \mathrel{\mathop{#3}\limits^{
    \vbox to#1{\kern-2\ex@
    \hbox{$\scriptstyle#2$}\vss}}}}
\makeatother

\newcommand{\ddepends}[2]{\oset[.5ex]{#1}{\xrightarrow{}}_{#2}}
\newcommand{\depends}[2]{\oset[.3ex]{#1}{\twoheadrightarrow_{#2}}}

\DeclarePairedDelimiter\abs{\lvert}{\rvert}
\makeatletter
\let\oldabs\abs
\def\abs{\@ifstar{\oldabs}{\oldabs*}}
\makeatother

\newcommand{\R} {\mathbb{R}}

\newcommand{\N} {\mathbb{N}}

\newcommand{\Q} {\mathbb{Q}}
\newcommand{\E} {\mathbb{E}}

\newcommand{\mom} {\ensuremath{\mathit{Mom}}}
\newcommand{\sens} {\ensuremath{\mathit{Sens}}}
\newcommand{\SRec} {\ensuremath{\mathit{SRec}}}
\newcommand{\MRec} {\ensuremath{\mathit{MRec}}}
\newcommand{\Eqs} {\ensuremath{\mathit{Eqs}}}

\newcommand{\seq}[3][{}]{\langle #2 \rangle_{#3}^{#1}}

\renewcommand{\paragraph}[1]{\par\smallskip \noindent\textit{#1}}

\begin{document}
\title{Automated Sensitivity Analysis for Probabilistic Loops}
\author{Marcel Moosbrugger \orcidlink{0000-0002-2006-3741} \inst{1} \and
Julian M\"ullner \orcidlink{0009-0006-2909-2297} \inst{1}\and
Laura Kov\'acs \orcidlink{0000-0002-8299-2714} \inst{1}}
\authorrunning{M. Moosbrugger et al.}
\institute{TU Wien, Vienna, Austria}
\maketitle              %
\begin{abstract}
We present an exact approach to analyze and quantify the sensitivity of higher moments of probabilistic loops with symbolic parameters, polynomial arithmetic and potentially uncountable state spaces.
Our approach integrates methods from symbolic computation, probability theory, and static analysis in order to automatically capture sensitivity information about probabilistic loops.
Sensitivity information allows us to formally establish how value distributions of probabilistic loop variables influence the functional behavior of loops, which can in particular be helpful when  choosing values of loop variables in order to ensure efficient/expected computations.
Our work uses algebraic techniques to model higher moments of loop variables via linear recurrence equations and introduce the notion of \emph{sensitivity recurrences}.
We show that sensitivity recurrences precisely model loop sensitivities, even in cases where the moments of loop variables do not satisfy a system of linear recurrences.
As such, we enlarge the class of probabilistic loops for which sensitivity analysis was so far feasible.
We demonstrate the success of our approach while analyzing the sensitivities of probabilistic loops. 

\end{abstract}

\section{Introduction}

Probabilistic programs are imperative programs enriched with the capability to draw from probability distributions.
By supporting native primitives to model uncertainty, probabilistic programming provides a powerful framework to model stochastic systems from many different areas, 
such as machine learning~\cite{Ghahramani2015}, biology~\cite{unsolvable}, cyber-physical systems~\cite{Chou2020,Selyunin2015}, cryptography~\cite{Barthe2012}, privacy~\cite{Barthe2012a}, and randomized algorithms~\cite{Motwani1995}.

A challenging task in the analysis of probabilistic programs comes from the fact that values, or even value distributions,  of symbolic parameters used within program expressions over probabilistic program variables are often unknown.
Sensitivity analysis aims to quantify how small changes in such parameters influence computation results \cite{Aguirre2021,BartheEGHS18}.
Sensitivity analysis thus provides additional information about the probabilistic program executions, even if some parameters are (partially) unknown. This sensitivity information can further be used, among others, in code optimization: sensitivity information quantifies the influence of parameters on the program variables, allowing to derive cost-effective estimates and optimize expected runtimes of probabilistic loops. 

The sensitivity analysis of probabilistic programs is however hard due to their intrinsic randomness:
program variables are no longer assigned single values but rather hold probability distributions \cite{Barthe2020}.
Uncountably infinite state spaces and non-linear assignments are further obstacles to the formal analysis of probabilistic programs.
In recent years, several frameworks to \emph{manually} reason about the sensitivity of probabilistic programs were proposed~\cite{Aguirre2021,BartheEGHS18,VasilenkoVB22}.
However, the state-of-the-art in \emph{automated} sensitivity analysis mainly focuses on loop-free programs such as \emph{Bayesian networks}~\cite{ChanD02,ChanD04,BartocciKS20A,StankovicBK22} and statically-bounded loops~\cite{HuangWM18}.
The technique presented in \cite{WangFCDX20} supports loops with variable-dependent termination times, but can only verify that the sensitivities obey certain bounds.
To the best of our knowledge, up to now, there is no automated and exact method supporting the sensitivity analysis of (potentially) unbounded probabilistic loops. 

\emph{In this paper, we propose a fully automatic technique for the sensitivity analysis of unbounded probabilistic loops.} The crux of our approach lies within the integration of
 methods from symbolic computation, probability theory and static analysis in order to automatically capture sensitivity information about probabilistic loops.
 Such an integrated framework allows us to also characterize a class of loops for which our technique is sound and complete.

\begin{figure*}[t]
    \begin{subfigure}[b]{0.47\textwidth}
        \centering
        \includegraphics[width=\linewidth]{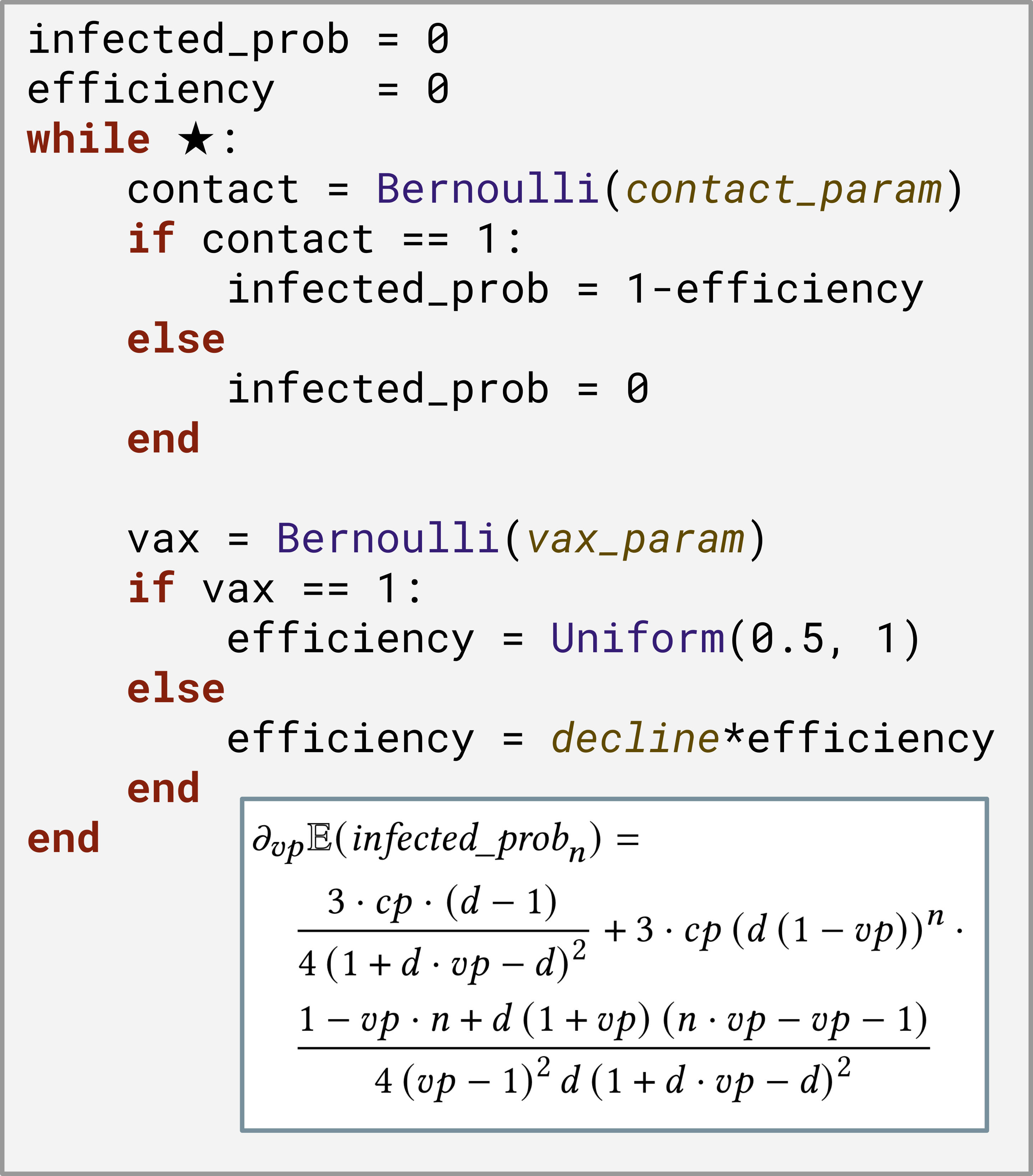}
        \caption{}
        \label{fig:running:1}
    \end{subfigure}
    \hfill %
    \begin{subfigure}[b]{0.4\textwidth} %
        \centering
        \includegraphics[width=\linewidth]{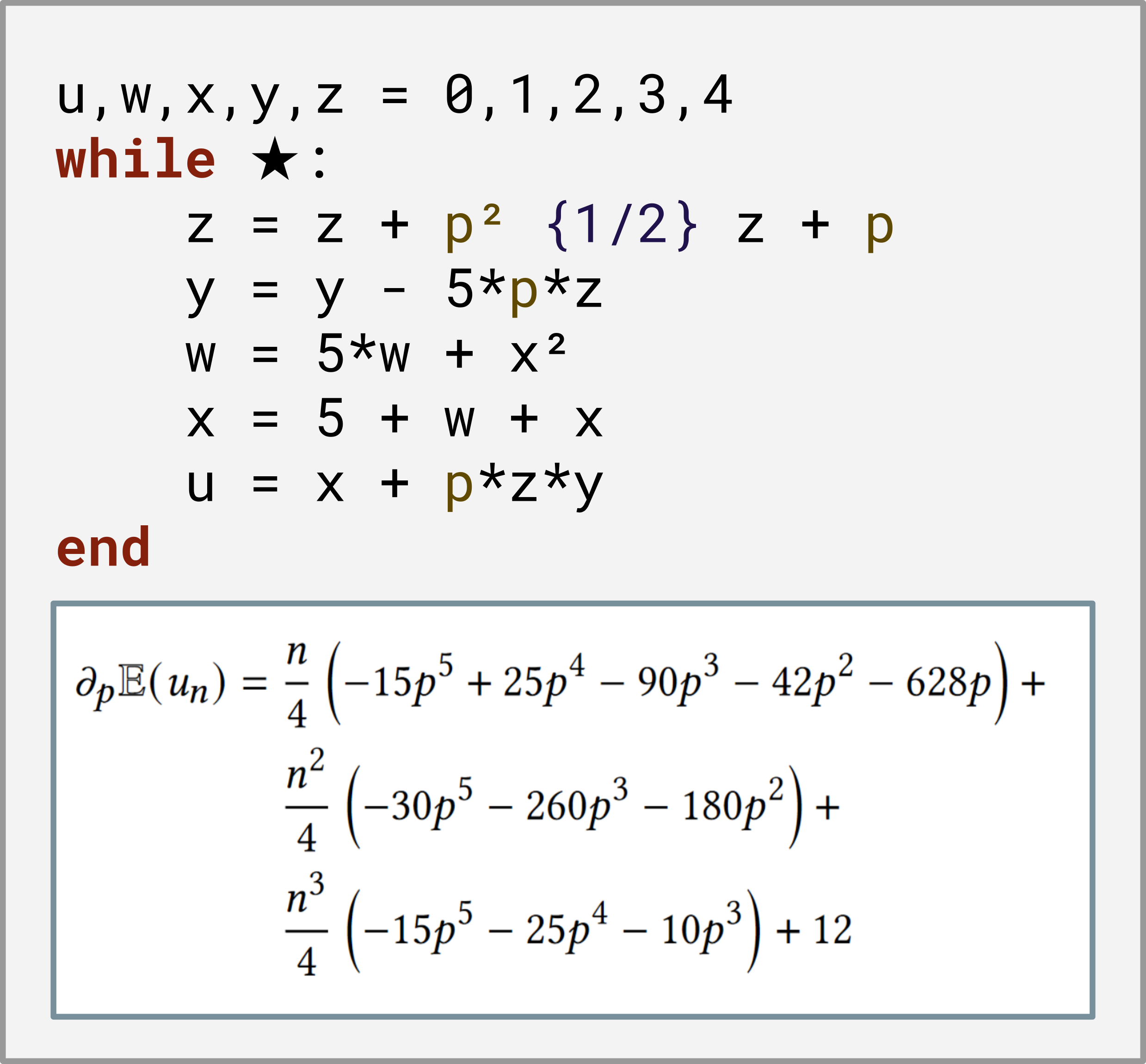}
        \caption{}
        \label{fig:running:2}
    \end{subfigure}
    \caption{Two examples of parameterized probabilistic loops, %
    where our approach automatically derives 
    loop sensitivities $\diffp$ as polynomial expressions depending on the loop counter $n$ and other parameters; for example \texttt{infected\_prob} with respect to \texttt{vax\_param} (Fig.~\ref{fig:running:1}) or that of \texttt{u} with respect to \texttt{p} (Fig.~\ref{fig:running:2}). 
    Using these results, our work shows that,  when assuming \texttt{decline=$0.9$}, \texttt{contact\_param=$0.7$}, after $n=10$ time steps and currently having \texttt{vax\_param=$0.1$}, then a small change $\varepsilon$ in \texttt{vax\_param} will decrease \texttt{infected\_prob} by approximately $1.7 \varepsilon$ in the next time step of Fig.~\ref{fig:running:1}. 
    }
    \label{fig:running-examples}
\end{figure*}

\noindent\paragraph{\bf Our framework for algebraic sensitivity analysis.} 
We advocate the use of algebraic recurrences to model the behavior of probabilistic loops. 
We combine and adjust techniques from symbolic summation, partial derivatives, and probability theory to provide a step towards the exact and automated sensitivity analysis of probabilistic loops, even in the presence of uncountable state spaces and polynomial assignments.
Figure~\ref{fig:running-examples} shows two probabilistic loops for which our work automatically computes the sensitivities of program variables with respect to different parameters. 
For example, Fig.~\ref{fig:running:1} depicts a probabilistic program, modelling the incidence of a disease within a population.
More precisely, it models the probability \texttt{infected\_prob} that a single organism within the population is infected, in dependence on symbolic parameters that model the amount of social interaction (\texttt{contac\_param}), the frequency of vaccinations (\texttt{vax\_param}) and  effect of a vaccination weakening over time (\texttt{decline}).
Sensitivity analysis helps to reason about the influence of these parameters on the disease infection process, answering for example the question ``How will an increase in the rate of vaccinations \texttt{vax\_param} influence the probability \texttt{infected\_prob} of an infection?''. Our work provides an algebraic approach to answering such and similar questions.

In a nutshell, our technique computes exact closed-form solutions for the sensitivities of (higher) moments of program variables for all, possibly infinitely many, loop iterations.
Higher moments are necessary to recover/estimate the value distributions of probabilistic loop variables and hence these moments help in inferring valuable sensitivity information for the variance or skewness.
In our work, we utilize algebraic techniques in probabilistic loop analysis to model moments of program variables with linear recurrences, so-called \emph{moment recurrences}~\cite{polar,BartocciKS19}.
However, {moment recurrences} do not support loops with intricate polynomial arithmetic, such as the loop in Figure~\ref{fig:running:2}.
To overcome this limitation, we propose the notion of \emph{sensitivity recurrences}, which shortcut computing closed-forms for variable moments and directly model sensitivities via linear recurrence equations.
In Figure~\ref{fig:running:2}, the program variable $w$ is independent of the parameter $p$.
By exploiting the independence of program variables from parameters, \emph{sensitivity recurrences} enable the exact sensitivity analysis for loops such as  Figure~\ref{fig:running:2}.
We characterize a class of probabilistic loops for which we prove \emph{sensitivity analysis via sensitivity recurrences} to be sound and complete.

\noindent\paragraph{\bf Our contributions.} We integrate symbolic computation, in particular symbolic summation and partial derivation, in combination with methods from probability theory into the landscape of probabilistic program reasoning. In particular, we argue that recurrence-based loop analysis yields a fully automated and precise way to derive sensitivity information over unknown symbolic parameters in probabilistic loops. As such, our paper brings the following main contributions:
\begin{itemize}
    \item We propose a fully automated approach for the sensitivity analysis of probabilistic loops based on \emph{moment recurrences} (Section~\ref{section:sensitvity:admissible}).
    \item We introduce \emph{sensitivity recurrences} and an algorithm for sensitivity analysis going beyond \emph{moment recurrences} (Section~\ref{section:sensitvity:non-admissible}, Algorithm~\ref{alg:sensitivity-recurrences}).
    \item We provide a precise characterization of the class of probabilistic loops for which \emph{sensitivity recurrences} are provably  sound and complete (Theorem~\ref{thm:sensitivity-computable}).
    \item We describe an experimental evaluation demonstrating the feasibility of our techniques on many interesting probabilistic programs  (Section~\ref{sec:evaluation}).
\end{itemize}

\section{Preliminaries \label{section:preliminaries}}
We write $\N$ for the natural numbers, $\R$ for the reals, $\overline{\Q}$ for the algebraic numbers, and $\mathbb{K}[x_1,\ldots, x_k]$ for the polynomial ring with coefficients in the field $\mathbb{K}$.
A polynomial consisting of a single monic term is a \emph{monomial}.
The expected value operator is denoted as $\E$.

\subsection{Syntax and Semantics of Probabilistic Loops}

\paragraph{Syntax.}
In this paper, we focus on unbounded probabilistic while-loops, as illustrated by the two examples of Figure~\ref{fig:running-examples} and introduced in \cite{polar}.
Our programming model considers non-nested while-loops preceded by a variable initialization part, with the loop body being a sequence of (nested) if-statements and variable assignments.
Unbounded probabilistic loops occur frequently when modeling dynamical systems. Guarded loops \texttt{while G: body} can be analyzed by considering the limiting behavior of unbounded loops of the form \texttt{while true: if G: body}.

The right-hand side of every variable assignment is either a probability distribution with existing moments (e.g. Normal or Uniform) and constant parameters, or a probabilistic choice of polynomials in program variables, that is $\texttt{x} = \mathit{poly}_1 \{p_1\} \dots \mathit{poly}_k \{p_k\}$, where \texttt{x} is assigned to $\mathit{poly}_i$ with probability $p_i$.
Further, programs can be parameterized by symbolic constants which represent arbitrary real numbers.
For further details, we refer to Appendix~\ref{appendix:syntax}.

Throughout this paper, we refer to programs from our programming  model simply by (probabilistic) loops or (probabilistic) programs.
For a program \prog{} we denote the set of program variables by \vars{} and the set of symbolic parameters by \params{}.

Dependencies between program variables is a syntactical notion introduced next, representing a central part in our work. 

\begin{definition}[Variable Dependency] \label{definition:variable-dependency}
Let \prog{} be a probabilistic loop and $x, y \in \vars{}$.
We say that $x$ \emph{depends directly on} $y$, and write  $x \ddepends{}{} y$, if $y$ appears in an assignment of $x$ or an assignment of $x$ occurs in an if-statement where $y$ appears in the if-condition.
Furthermore, we say that the dependency is \emph{non-linear}, denoted as $x \ddepends{N}{} y$, if $y$ appears non-linearly in an assignment of $x$.
\end{definition}

By $\depends{}{}$ we denote the transitive closure of $\ddepends{}{}$.
Regarding non-linearity, we write $x \depends{N}{} y$, if at least one of the direct dependencies from $x$ to $y$ is non-linear.

\begin{example}
In Figure~\ref{fig:running:2}, we have (among others) $y \ddepends{}{} z$, $w \ddepends{N}{} x$, $u \depends{N}{} w$, and $w \depends{N}{} u$.
To illustrate the influence of if-conditions, in Figure~\ref{fig:running:1}, note that $\textit{efficiency} \ddepends{}{} \textit{vax}$ and $\textit{infected\_prob} \depends{}{} \textit{vax}$.
\end{example}

\paragraph{Semantics.}
Operationally, every probabilistic loop models an infinite-state Markov chain,  which in turn induces a canonical probability space.
Due to brevity, we omit the straightforward but rather technical construction of the Markov chains associated to probabilistic loops.
For more details, we refer the interested reader to \cite{polar,Durrett2019}.
For an arithmetic expression $\textit{Expr}$ in program variables, we denote by $\textit{Expr}_n$ the stochastic process evaluating \textit{Expr} after the $n$th loop iteration.
 
\subsection{C-finite Recurrences}\label{section:preliminaries:c-finite}

We  recall  notions from algebraic recurrences~\cite{everest2003recurrence,kauers2011concrete}, adjusted to our work. 

A \emph{sequence} of algebraic numbers is a function $u \colon\N\to\overline\Q$, succinctly denoted by $\seq[\infty]{u(n)}{n=0}$ or $\seq{u(n)}{n}$.
A \emph{recurrence} for the sequence $u$ of order $\ell\in\N$ is specified by a function $\textrm{f} : \R^{\ell+1} \to \R$ and given by the equation $u(n{+}\ell) = \textrm{f}(u(n{+}\ell{-}1), \ldots, u(n{+}1), u(n),n)$.
The \emph{solutions} of a recurrence are the sequences satisfying the recurrence equation.
Of particular relevance to our work is the class of \emph{linear recurrences with constant coefficients} or more shortly, \emph{C-finite recurrences}.
The sequence $u$ satisfies a C-finite recurrence if $u(n{+}\ell) = c_{\ell{-}1} u(n{+}\ell{-}1) + c_{\ell{-}2} u(n{+}\ell{-}2) + \cdots + c_{0} u(n)$ holds,
where  $c_0,\ldots, c_{\ell-1}\in \overline\Q$ are constants and $c_0 \neq 0$.
Every C-finite recurrence is associated with its \emph{characteristic polynomial}
$x^{n} - c_{\ell-1} x^{\ell-1} - \cdots - c_{1} x - c_{0}$.
The solutions of C-finite recurrences can always be computed~\cite{kauers2011concrete} and written in closed-form as \emph{exponential polynomials}.
More precisely, if $\seq{u(n)}{n}$ is the solution to a C-finite recurrence, then $u(n) = \sum_{k=1}^r P_k(n) \lambda_k^n$ where $P_k(n)\in\overline{\Q}[n]$ and $\lambda_1,\ldots, \lambda_r$ are  the roots of the characteristic polynomial.
The properties of C-finite recurrences also hold for systems of C-finite recurrences (systems of linear recurrence equations with constant coefficients, specifying multiple sequences).

\subsection{Higher Moment Analysis using Recurrences} \label{section:preliminaries:moment-recurrences}

For a random variable $x$, its higher moments are defined as $\E(x^k)$ for $k \in \N$.
More generally, mixed moments for a set of random variables $S$ are expected values of monomials in $S$.
Recent works in probabilistic program analysis \cite{polar,BartocciKS20} introduced techniques and tools based on C-finite recurrences to compute higher moments of program variables for probabilistic loops.
For example, for a probabilistic loop, $k \in \N$ and a program variable $x$, a closed-form solution for the $k$th higher moment of $x$ parameterized by the loop iteration $n$, that is $\E(x^k_n)$, is computed in~\cite{polar} using the \Polar{} tool. 
This is achieved by first normalizing the program to eliminate if-statements and ensure every variable is only assigned once in the loop body.
Then, a system of C-finite recurrences is constructed that models expected values of monomials in program variables.
More precisely, for a monomial $M$ in program variables, the work of~\cite{polar} constructs a linear recurrence equation, relating the expected value of $M$ in iteration $n{+}1$ to the expected values of program variable monomials in iteration $n$.
The linear recurrence for the expected value of $M$ in iteration $n{+}1$ is constructed by starting with the expression $\E(M_{n{+}1})$ and replacing variables contained in the expression by their assignments bottom-up as they appear in the loop body.
Throughout, the linearity of expectation is used to convert expected values of polynomials into expected values of monomials (cf. Appendix~\ref{appendix:moment-recurrences}).

We adopt the setting of~\cite{polar,BartocciKS20} and refer by \emph{moment recurrences} to the recurrence equations these techniques construct for moments of program variables.

\begin{definition}[Moment Recurrence]\label{def:moment-rec}
Let \prog{} be a probabilistic loop and $M$ a monomial in \vars{}.
A \emph{moment recurrence} for $M$ is an equation $\E(M_{n+1}) = \sum_{i=1}^r c_i \cdot \E(W^{(i)}_n)$
where $c_i \in \overline\Q$ and all $W^{(i)}$ are monomials in \vars{}.
\end{definition}

In order to compute a closed-form solution for $\E(x^k_n)$, we employ~\cite{polar} to first compute a moment recurrence $R$ for the monomial $x^k$. Next, we derive moment recurrences for all monomials $W^{(i)}$ in $R$ (cf. Definition~\ref{def:moment-rec}) to construct a system of C-finite recurrences.

\begin{example}\label{example:recurrences:running:1}
Consider the program from Figure~\ref{fig:running:1}.
For a more succinct representation, we abbreviate the symbolic parameters as
$cp := \textit{contact\_param}$; $vp := \textit{vax\_param}$ and $d := \textit{decline}$.
The first moments of the program variables are modeled through the following system of C-finite recurrences~\cite{polar}:
\begin{flalign*}
    \E(\textit{infected_prob}_{n+1}) &= cp - cp \cdot \E(\textit{efficiency}_n) \\
    \E(\textit{efficiency}_{n+1}) &= (d - d \cdot vp) \cdot \E(\textit{efficiency}_n) + \frac{3}{4} \cdot vp
\end{flalign*}
The initial values of $\E(\textit{infected_prob}_{n})$ and $\E(\textit{efficiency}_{n})$ are both $0$.
The system can be automatically solved \cite{kauers2011concrete} to obtain closed-form solutions, which are, when expanded, exponential polynomials, e.g. for $\E(\textit{infected_prob}_{n})$:
\begin{align*}
    \E(\textit{infected_prob}_{n}) &= cp + \frac{3 \cdot vp \cdot cp \cdot \left(\left(d - d\cdot vp\right)^{n-1} - 1\right)}{4 \left(d \cdot vp - d + 1\right)} %
\end{align*}
\end{example}

We note that moment recurrences do not always exist.
Moreover, termination is not guaranteed when recursively inferring the moment recurrences for all monomials $W^{(i)}$ in Definition~\ref{def:moment-rec} in order to construct a C-finite system.

\begin{example} \label{example:running-example-non-moment-computable}
To illustrate that the approach based on moment recurrences does not work unconditionally, consider the loop from Figure~\ref{fig:running:2} and construct the moment recurrence $\E(w_{n+1}) = 5 \cdot \E(w_n) + \E(x_n^2)$.
Since the recurrence contains $\E(x_n^2)$, we require the moment recurrence $\E(x_{n+1}^2) = \E((5+w_{n+1}+x_n)^2) = \E(w_{n+1}^2) + \ldots$ which requires the recurrence for $\E(w_n^2)$.
This in turn necessitates a recurrence for $\E(x_n^4)$, which necessitates the recurrence for $\E(w_n^4)$ and so on.
This process will repeatedly require recurrences for increasing moments of $x_n$ and $w_n$, implying that this process will not terminate.
\end{example}

To circumvent variable dependencies and compute closed-forms of moment recurrences, we note that the following two conditions on the probabilistic loops ensure  existence and computability of higher order moments.

\begin{definition}[Admissible Loop]\label{def:admissible-loops}
A loop is \emph{admissible} if
\begin{enumerate}
    \item all variables in branching conditions only assume values in a finite set (i.e. they are \emph{finite valued}), and
    \item\label{moment-comp:self} no variable $x$ is non-linearly self-dependent ($x \not \depends{N}{} x$)
    \footnote{
    While~\cite{polar} allows arbitrary dependencies among finite valued variables, our work  omits this generalization for simplicity. 
    Nevertheless,  our results also apply to admissible loops with  arbitrary dependencies among finite valued variables.
    }.
\end{enumerate}
\end{definition}

\begin{example}
The probabilistic loop in Figure~\ref{fig:running:1} is admissible.
However, the program in Figure~\ref{fig:running:2} is not admissible.
It does not satisfy condition~\ref{moment-comp:self}: the variable $x$ depends linearly on $w$ and $w$ depends quadratically on $x$; therefore, $x$ is non-linearly self-dependent.
\end{example}

Admissible probabilistic loops are \emph{moment-computable}~\cite{polar}, that is,   higher moments of  program variables admit computable closed-forms as exponential polynomials.
The restriction on finite valued variables in branching conditions is necessary to guarantee computability and completeness: a single branching statement involving an unbounded variable renders the program model Turing-complete~\cite{polar}.

\section{Sensitivity Analysis}\label{section:sensitivity}

In this section, we study the sensitivity of program variable moments with respect to symbolic parameters. 
We present two exact and fully automatic methods to answer the question of how small changes in symbolic parameters influence the moments of program variables. As such, we exploit the fact that closed-forms for variable moments in admissible loops are computable (Section~\ref{section:sensitvity:admissible}). We further go beyond the admissible loop setting (Section~\ref{section:sensitvity:non-admissible}) and devise a sensitivity analysis technique applicable to some non-admissible loops, such as the program in Figure~\ref{fig:running:2}.

\begin{definition}[Sensitivity]\label{def:sensitivity}
Let \prog\ be a probabilistic loop, $x \in \vars$ and $p \in \params$.
The \emph{sensitivity} of the $k$th moment of $x$ with respect to $p$, denoted as  $\diffp \E(x_n^k)$, is defined as the partial derivative of $\E(x_n^k)$ with respect to $p$, and parameterized by loop counter $n$.
For monomials $M$ of variables, the sensitivity $\diffp \E(M_n)$ is defined analogously.
\end{definition}

Similar to \emph{moment computability}~\cite{polar}, we define a program to be \emph{sensitivity computable} if the sensitivities of all the variables' expected values are expressible in closed-form.

\begin{definition}[Sensitivity Computability]\label{def:sensitivity-computable}
Let \prog{} be a probabilistic program and $p \in \params$.
\prog{} is \emph{sensitivity computable with respect to $p$}, if for every variable $x \in \vars$ the sensitivity $\diffp \E(x_n)$ has an exponential polynomial closed-form that is computable.
\end{definition}

\subsection{Sensitivity Analysis for Admissible Loops}\label{section:sensitvity:admissible}
As mentioned in Section~\ref{section:preliminaries}, for admissible loops, any moment of every program variable admits a closed-form solution as an exponential polynomial which is computable.
That is,  for a program variable $x$ and $k \in \N$, the $k$th moment of $x$ can be written as
$\E(x_n^k) = \sum_{j=0}^r P_j(n) \lambda_j^n$,
where $P_j \in \overline{\Q}[n]$ and $\lambda_j \in \overline{\Q}$ may contain symbolic parameters.
We next show that based on the closed-forms of variable moments, we can compute exponential polynomials representing the sensitivities of moments on parameters.

\begin{theorem}[Admissible Sensitivities]\label{thm:admissible-diff}
Let \prog\ be an admissible program, $x \in \vars$, $p \in \params$, and $k \in \N$.
Then, the sensitivity $\diffp \E(x_n^k)$ has an exponential polynomial closed-form that is computable.
\end{theorem}
\begin{proof}
Because \prog{} is admissible, $\E(x^k_n)$ can be expressed as an exponential polynomial.
We show that the sensitivity can be expressed as an exponential polynomial by expanding $\E(x_n^k)$ into a sum of exponential monomials:
$\E(x_n^k) = \sum_{j=0}^r P_j(n) \lambda_j^n = \sum_{j=0}^r \sum_{i=0}^{m_j} M_{ij}(n) \lambda_j^n,$
where $m_j$ is the number of monomials in $P_j$ and every $M_{ij}$ is a monomial.
Note that every $M_{ij}$ and $\lambda_j$ may depend on the symbolic constant $p$.
The derivative of the exponential monomials can then be obtained by applying the product rule for derivatives:
\begin{align*}
\diffp \E(x_n^k) &= \sum_{j=0}^r \sum_{i=0}^{m_j} (\diffp M_{ij}(n)) \lambda_j^n + M_{ij}(n) \cdot n \cdot (\diffp \lambda_j) \cdot \lambda_j^{n-1} \\
&= \sum_{j=0}^{r} (\diffp P_j(n) + P_j(n) \cdot n \cdot \diffp \lambda_j \cdot \frac{1}{\lambda_j}) \lambda_j^{n}
\end{align*}
It is left to show that the exponential polynomial $\diffp \E(x^k_n)$ is computable.
Because \prog{} is admissible, an exponential polynomial for $\E(x^k_n)$ is computable.
Now, the second claim follows from the fact that exponential polynomials are elementary and that the derivative of any elementary function is computable.\qed
\end{proof}

As a corollary, admissible loops are sensitivity computable.
Although \emph{sensitivity computability} only refers to first moments, Theorem~\ref{thm:admissible-diff} shows that for admissible loops, sensitivities of \emph{all} higher moments of program variables admit a computable closed-form.

\begin{example}\label{example:sensitivities:running:1}
Consider Figure~\ref{fig:running:1}.
In Example~\ref{example:recurrences:running:1} we stated the closed-form solutions of $\E(\textit{infected_prob}_{n})$.
The sensitivities of the respective expected values can be computed by symbolic differentiation and, by Theorem~\ref{thm:admissible-diff}, can be expanded to exponential polynomials. 
For example, the following expression describes the sensitivity of $\E(\textit{infected_prob}_{n})$ with respect to the parameter $vp$:

\noindent
\resizebox{\linewidth}{!}{
\begin{minipage}{\linewidth}
\begin{flalign*}
    \partial_{vp} \E(\textit{infected_prob}_{n}) = & \frac{3 \cdot cp  \left(1- vp \cdot n + d \left(1+vp\right) \left(n \cdot vp - vp -1\right) \right)  \left(d \left(1-vp\right)\right)^n}{4 \left(vp-1\right)^2 d \left(1 + d \cdot vp - d\right)^2} \\
    & \quad + \frac{3 \cdot cp \cdot \left(d-1\right)}{4 \left(1+d \cdot vp - d \right)^2}%
\end{flalign*}
\end{minipage}}
\end{example}

\subsection{Sensitivity Analysis for Non-Admissible Loops}\label{section:sensitvity:non-admissible}

In general, moments of program variables of non-admissible loops do not satisfy linear recurrences.
Therefore, we cannot utilize closed-forms of the moments for sensitivity analysis.
Nevertheless, sensitivity analysis is feasible even for some non-admissible loops.
In this section, we propose a novel sensitivity analysis approach applicable to non-admissible loops.
Moreover, we characterize the class of (non-admissible) loops for which our method is sound and complete.

For admissible loops, linear recurrences describing variable moments can be used as an intermediary step to compute sensitivities.
The core of our approach towards handling  non-admissible loops is to shortcut moment recurrences and devise recurrences directly for sensitivities.
Due to independence with respect to the sensitivity parameter, sensitivities of program variables can follow a linear recurrence even though their moments do not.
We illustrate the idea of our new method on the non-admissible loop from Figure~\ref{fig:running:2}.

\begin{example}\label{ex:sensitivity-rec}
Consider the non-admissible program from Figure~\ref{fig:running:2}.
The moment recurrences for all program variables are:
\begin{alignat*}{2}
    \E(z_{n+1}) &= \E(z_n) + 0.5 \cdot (p + p^2) & \E(y_{n+1}) &= \E(y_n) - 5p \cdot \E(z_{n+1})\\
    \E(w_{n+1}) &= 5 \cdot \E(w_n) + \E(x_n^2) & \E(x_{n+1}) &= 5 + \E(w_{n+1}) + \E(x_n)\\ 
    \E(u_{n+1}) &= \E(x_{n+1}) + p \cdot \E(zy_{n+1}) & &
\end{alignat*}
As illustrated in Example~\ref{example:running-example-non-moment-computable}, we cannot complete the recurrences to a C-finite system because both $w$ and $x$ are non-linearly self-dependent.
Therefore, we cannot compute closed-form solutions for $\E(w_n)$ and $\E(x_n)$.
However, we can shortcut solving for $\E(w_n)$ and $\E(x_n)$ by differentiating the moment recurrences with respect to $p$ and establish recurrences directly for the sensitivities:
\begin{align*}
    \diffp \E(z_{n+1}) &= \diffp \E(z_n) + 0.5 \cdot (1 + 2p)\\
    \diffp \E(y_{n+1}) &= \diffp \E(y_n) - 5p \cdot \diffp \E(z_{n+1}) - 5 \cdot \E(z_{n+1})\\
    \diffp \E(w_{n+1}) &= 5 \cdot \diffp \E(w_n) + \diffp \E(x_n^2)\\
    \diffp \E(x_{n+1}) &= \diffp \E(w_{n+1}) + \diffp \E(x_n)\\ 
    \diffp \E(u_{n+1}) &= \diffp \E(x_{n+1}) + \E(zy_{n+1}) + p \cdot \diffp \E(zy_{n+1})
\end{align*}
Now, because the variables $w$ and $x$ do not depend on the parameter $p$, we conclude that $\diffp \E(w_n) \equiv \diffp \E(x_n) \equiv 0$.
The sensitivity recurrences thus simplify:
\begin{align*}
    \diffp \E(z_{n+1}) &= \diffp \E(z_n) + \frac{1 + 2p}{2}\\
    \diffp \E(y_{n+1}) &= \diffp \E(y_n) - 5p \cdot \diffp \E(z_{n+1}) - 5 \cdot \E(z_{n+1})\\
    \diffp \E(u_{n+1}) &= \E(zy_{n+1}) + p \cdot \diffp \E(zy_{n+1})
\end{align*}
We can interpret sensitivities such as $ \diffp \E(z_{n})$ or $\diffp \E(u_{n})$ as atomic recurrence variables.
In the resulting recurrences, all variables with non-linear self-dependencies vanished.
Therefore, the recurrences can be completed to a C-finite system and solved by existing techniques, even though $\E(w_n)$ and $\E(x_n)$ are not C-finite.
The resulting system of recurrences consists of all recurrences for sensitivities and moments that appear on the right-hand side of another recurrence.
That is, the system of recurrences consists of the sensitivity recurrences for $\diffp \E(z)$, $\diffp \E(y)$, $\diffp \E(u)$, $\diffp \E(yz)$, $\diffp \E(z^2)$ and the moment recurrences for $\E(z)$, $\E(y)$, $\E(yz)$, $\E(z^2)$.
\end{example}

Motivated by Example~\ref{ex:sensitivity-rec}, we introduce the notion of \emph{sensitivity recurrences}.

\begin{definition}[Sensitivity Recurrence]\label{def:sensitivity-recurrences}
Let \prog\ be a program, $p \in \params$ a symbolic parameter, $M$ a monomial in $\vars$ and let $\E(M_{n+1}) = \sum_{i=1}^r c_i \cdot \E (W^{(i)}_n)$ be the moment-recurrence of $M$.
Then the \emph{sensitivity recurrence} of $M$ with respect to $p$ is defined as
\begin{flalign}\label{eq:sensitivity-recurrence}
\begin{split}
\boxed{\diffp \E(M_{n+1})} & \coloneqq \frac{\partial \E \left(M_{n+1}\right)}{\partial p} 
=
\frac{\partial}{\partial p} \left( \sum_{i=1}^r c_i \cdot \E \left(W^{(i)}_n\right) \right) \\
& \boxed{ = \sum_{i=1}^r  \left(\frac{\partial}{\partial p} c_i \right) \cdot \E \left(W^{(i)}_n\right) + c_i \cdot \diffp \E \left(W^{(i)}_n\right)}
\end{split}
\end{flalign}
\end{definition}

\begin{algorithm}
\caption{Computing Sensitivities via Sensitivity Recurrences}\label{alg:sensitivity-recurrences}
\begin{algorithmic}[1]
\Require program \prog{}, monomial $M$ in \vars{}, $p \in \params$
\Ensure closed-form for $\diffp \E(M_n)$
\If{$M$ is $p$-independent}
\State \Return $0$\label{algo:step:return-independent}
\EndIf
\State $\Eqs \leftarrow \emptyset$; $\mom \leftarrow \emptyset$; $\sens \leftarrow \{ M \}$

\color{algo1}
\While{$\sens \neq \emptyset$} \Comment{Add all necessary sensitivity recurrences}
    \State pick $W \in \sens$; $\sens \leftarrow \sens \setminus \{W\}$\label{algo:step:pick1}
    \State $\SRec \leftarrow $ sensitivity recurrence of $W$
    \State Replace every $\diffp \E(W'_n)$ in $\SRec$ by $0$ if $W'$ is $p$-independent\label{algo:step:var-independence}
    \State Replace every $(\nicefrac{\partial}{\partial p} c) \E(W'_n)$ in $\SRec$ by $0$ if $(\nicefrac{\partial}{\partial p} c) = 0$\label{algo:step:const-independence}
    \State $\Eqs \leftarrow \Eqs \cup \{ \SRec \}$

    \State Add to $\sens$ all monomials $W'$ s.t. $\diffp \E(W'_n)$ in $\SRec$\label{algo:step:add-sens}
    \State $\hookrightarrow$ and the sensitivity recurrence of $W'$ $\not\in \Eqs$

    \State Add to $\mom$ all monomials $W'$ s.t. $\E(W'_n)$ in $\SRec$\label{algo:step:add-mom1}
    
\EndWhile
\color{algo2}
\While{$\mom \neq \emptyset$} \Comment{Add all necessary moment recurrences}
    \State pick $W \in \mom$; $\mom \leftarrow \mom \setminus \{W\}$\label{algo:step:pick2}
    \State $\MRec \leftarrow $ moment recurrence of $W$
    \State $\Eqs \leftarrow \Eqs \cup \{ \MRec \}$
    
    \State Add to $\mom$ all monomials $W'$ s.t. $\E(W'_n)$ in $\MRec$\label{algo:step:add-mom2}
    \State $\hookrightarrow$ and the moment recurrence of $W' \not\in \Eqs$
    
\EndWhile
\color{black}
\State $S \leftarrow$ solve system of C-finite recurrences $\Eqs$
\State \Return closed-form of $\diffp \E(M_n)$ from $S$
\end{algorithmic}
\end{algorithm}

The sensitivity recurrence of $M$ equates the sensitivity of $M$ at iteration $n{+}1$ to moments \emph{and sensitivities} at iteration $n$.
Along the ideas in Example~\ref{ex:sensitivity-rec}, we provide with Algorithm~\ref{alg:sensitivity-recurrences} a procedure for sensitivity analysis also applicable to non-admissible loops.
The idea of Algorithm~\ref{alg:sensitivity-recurrences} is to determine $\diffp \E(M_{n})$ by constructing a C-finite system consisting of all necessary recurrence equations for the moments and sensitivities of program variables.
As illustrated in Example~\ref{ex:sensitivity-rec}, we can exploit the independence of variables from the sensitivity parameter $p$ to simplify the problem:
if a monomial $W'$ is independent from $p$ then $\diffp \E (W'_n) \equiv 0$.
Moreover, if $p$ does not appear in the constant $c_i$ of Equation~\eqref{eq:sensitivity-recurrence}, then $(\nicefrac{\partial}{\partial p}) c_i = 0$, and hence the moment recurrence of $W'$ does not need to be constructed (lines~\ref{algo:step:var-independence}--\ref{algo:step:const-independence} of Algorithm~\ref{alg:sensitivity-recurrences}).
This is essential if the \emph{expected value} of $W'$ does not admit a closed-form.
Algorithm~\ref{alg:sensitivity-recurrences} is sound by construction, however, termination is non-trivial.
In the remainder of this section, we formalize the notion of parameter (in)dependence and give a characterization of the class of non-admissible loops for which Algorithm~\ref{alg:sensitivity-recurrences} terminates. As a consequence of Algorithm~\ref{alg:sensitivity-recurrences}, we show that  sensitivity recurrences yield an exact and complete technique for sensitivity analysis (Theorem~\ref{thm:sensitivity-computable}).

\begin{definition}[$p$-Dependent Variable]\label{def:p-dependency}
Let \prog{} be a program with parameter $p \in \params$. A variable $x \in \vars$ is \emph{$p$-dependent}, if (1) $p$ appears in an assignment of $x$, (2) $x$ depends on some $y \in \vars$ ($x \depends{}{} y$) and $y$ is $p$-dependent or (3) an assignment of $x$ occurs in an if-statement where $p$ appears in the if-condition.
A variable is \emph{$p$-independent} if it is not $p$-dependent.
A monomial $M$ in program variables is $p$-dependent if $M$ contains at least one $p$-dependent variable, otherwise it is $p$-independent.
\end{definition}

For any $p$-independent monomial $M$ in program variables,  the corresponding sensitivity $\diffp \E(M_n)$ is zero (by using induction on $n$ and applying Definition~\ref{def:p-dependency}).

\begin{lemma} \label{lemma:p-indep-sensitivity-is-zero}
Let \prog\ be a program, $p \in \params$ a symbolic parameter and $M$ a $p$-independent monomial in $\vars$, then it holds that the sensitivity variable of $M$ is zero, i.e., $\forall n \geq 0: \diffp \E(M_n) = 0$.
\end{lemma}

In Example~\ref{ex:sensitivity-rec}, the moments $\E(w_n)$ and $\E(x_n)$ do not admit closed-forms.
We resolved this issue by differentiating all moment recurrences and working directly with the sensitivity recurrences, where the moment recurrences for $w$ and $x$ vanished.
Crucial for this phenomenon is the fact that the variables $w$ and $x$ are independent of the sensitivity parameter $p$.

However, a second fact is necessary to guarantee that the moment recurrences of $w$ and $x$ do not appear in the resulting system of recurrences:
Assume some new variable $v$ depends on $x$ and has the moment recurrence $\E(v_{n+1}) = \E(v_n) + p \cdot \E(x_n)$.
Then the sensitivity recurrence for $v$ is given by $\diffp \E(v_{n+1}) = \diffp \E(v_n) + \E(x_n) + p \cdot \diffp \E(x_n)$.
Even though $x$ itself is $p$-independent, $\E(x_n)$ remains in the sensitivity recurrence of $v$ because the coefficient of $\E(x_n)$ contains the parameter $p$.
A similar effect occurs if the moment recurrence for $v$ was $\E(v_{n+1}) = \E(v_n) + \E(z_n x_n)$, because $z$ is $p$-dependent.

Our goal is to characterize the class of probabilistic loops for which sensitivity recurrences yield a sound and complete method for sensitivity analysis.
Hence, we need to capture the notion that some dependencies between variables are free of multiplicative factors involving the sensitivity parameter.
We do this in the following definition by refining our dependency relation $\depends{}{}$.

\begin{definition}[$p$-Influenced Dependency]\label{definition:p-infl-dependency}
Let \prog{} be a program with parameter $p \in \params$ and $x,y \in \vars$ with $x \ddepends{}{} y$.
Then, the direct dependency between $x$ and $y$ is \emph{$p$-influenced}, written as $x \ddepends{}{p} y$, if at least one of the following conditions hold:

\begin{itemize}
    \item An assignment of $x$ contains $y$ and occurs in an if-statement with the if-condition involving $p$ or a $p$-dependent variable.
    \item An assignment of $x$ contains $y$ and is a probabilistic choice with some probability of the choice depending on $p$.
    \item An assignment of $x$ contains a term $c \cdot M \cdot y$ where $c$ is constant and $M$ is a monomial in program variables (possibly containing $y$). Moreover, either $c$ contains $p$ or $M$ contains a $p$-dependent variable.
\end{itemize}

If $x \depends{}{} y$, we write $x \depends{}{p} y$ if some dependency from $x$ to $y$ is $p$-influenced.
If $x \depends{}{} y$ and $x \not\depends{}{p} y$ we call the dependency between $x$ and $y$ \emph{$p$-free}.
\end{definition}

Definition~\ref{definition:p-infl-dependency} covers all cases in the construction of moment recurrences that introduce multiplicative factors depending on the sensitivity parameter $p$ \cite{polar}.
We provide details on the construction of moment recurrences in  Appendix~\ref{appendix:moment-recurrences}.

More concretely, assume $\prog{}$ to be a program and $x \in \vars{}$.
The moment recurrence of $x$ contains expected values of monomials $M$ of program variables.
Additionally, the moment recurrences of any $M$ will again contain expected values of monomials of program variables and so on.
We capture all of these monomials with the notion of \emph{descendant monomials} in Definition~\ref{def:descendant-monomial}.
Intuitively, to construct a system of moment recurrences for $\E(x_n)$ one needs to include the moment recurrences of all descendants of $x$.

\begin{definition}[Descendant Monomial]\label{def:descendant-monomial}
Let \prog{} be a program, $x \in \vars$, and $M$ a monomial in program variables.
The monomial $M$ is a \emph{descendant} of the variable $x$ if (1) $M = x$, or (2) $M$ occurs in the moment recurrence of a monomial $W$ and $W$ is a descendant of $x$.
The variable $x$ is an \emph{ancestor} of $M$.
\end{definition}

There is a dependency between $x$ and any variable of any descendant of $x$, which means $x \depends{}{} y$ for every descendant $M$ of $x$ and every variable $y$ in $M$.
Our dependency relation from Definition~\ref{definition:p-infl-dependency} allows us to pinpoint the variables in the moment recurrence of any descendant of $x$ (Definition~\ref{def:descendant-monomial}) with a multiplicative factor involving the sensitivity parameter.
Definitions~\ref{definition:p-infl-dependency} and \ref{def:descendant-monomial} together with the procedure constructing moment recurrences (Appendix~\ref{appendix:moment-recurrences}) yield:

\begin{lemma}[$p$-Influenced Moment Recurrence]\label{lemma:p-inf-moment-rec}
Let \prog\ be a program, $x \in \vars$, and $p \in \params$.
Assume $M$ is a monomial in program variables descending from $x$.
Let $W$ be a monomial in $M$'s moment recurrence with non-zero coefficient $c$.
If the parameter $p$ occurs in $c$, then for all variables $y$ in $W$ we have $x \depends{}{p} y$.
Moreover, if some variable $z$ in $W$ is $p$-dependent, then for all variables $y$ in $W$ different from $z$ we have $x \depends{}{p} y$.
\end{lemma}

We now state our main result (Theorem~\ref{thm:sensitivity-computable}) describing the class of probabilistic loops for which Algorithm~\ref{alg:sensitivity-recurrences} terminates and, hence, sensitivity recurrences are sound and complete.
We characterize the class of loops in terms of our dependency relations as well as variables with non-linear self-dependencies, which we refer to as \emph{defective} variables.

\begin{definition}[Defective Variables]
Let \prog\ be a program and $x \in \vars$, then $x$ is \emph{defective} if $x \depends{N}{} x$.
Otherwise, $x$ is \emph{effective}.
\end{definition}

\begin{restatable}[Non-Admissible Sensitivities]{theorem}{sensitivitycomputable}\label{thm:sensitivity-computable}
Let \prog{} be a probabilistic program, $p \in \params$, $x \in \vars$, and assume all the following conditions:
\begin{enumerate}
    \item\label{thm:sensitivity-computable:finite} All variables occuring in branching conditions are finite.
    \item\label{thm:sensitivity-computable:independent} All defective variables are $p$-independent.
    \item\label{thm:sensitivity-computable:p-free} All dependencies on defective variables are $p$-free.
\end{enumerate}
Then, for every monomial $M$ in program variables descending from $x$, Algorithm~\ref{alg:sensitivity-recurrences} terminates on input \prog{}, $M$ and $p$.
\end{restatable}
\begin{proof}[Sketch]
Algorithm~\ref{alg:sensitivity-recurrences} does not terminate iff infinitely many monomials are added to the set $\sens$ one line~\ref{algo:step:add-sens} or to the set $\mom$ on lines \ref{algo:step:add-mom1} or \ref{algo:step:add-mom2}.
However, every monomial added to these sets decreases with respect to some well-founded ordering.
Hence, only finitely many monomials are added and Algorithm~\ref{alg:sensitivity-recurrences} terminates.
This holds by using a well-founded ordering for monomials of effective variables and showing that all monomials added to $\sens$ or $\mom$ do \emph{not} contain defective variables.
Assuming then that some monomial added to the sets $\sens$ or $\mom$ contains defective variables  leads to contradictions using conditions~\ref{thm:sensitivity-computable:independent} and ~\ref{thm:sensitivity-computable:p-free}, and Lemma~\ref{lemma:p-inf-moment-rec}.
See Appendix~\ref{appendix:main-proof} for more details. 
\qed
\end{proof}

Theorem~\ref{thm:sensitivity-computable} characterizes the class of probabilistic loops for which sensitivity recurrences provide a sound and complete method for sensitivity analysis.
As an immediate corollary, this class of loops is sensitivity computable because every variable is a descendant of itself.
Note that all conditions of Theorem~\ref{thm:sensitivity-computable} are statically checkable:
the concepts of defective variables, $p$-independent variables, and $p$-free dependencies are purely syntactic notions.
Moreover, program variables occurring in branching conditions only admitting finitely many values can be verified using standard techniques based on \emph{abstract interpretation}.

Theorem~\ref{thm:sensitivity-computable} also applies to sensitivity analysis for higher moments: let $v \in \vars$ and $k \in \N$, then Theorem~\ref{thm:sensitivity-computable} covers the sensitivity of $v$'s $k$th moment if $v^k$ is a descendant of some variable.
Otherwise, $v^k$ can be dealt with by introducing a fresh variable $w$ and appending the assignment $w := v^k$ to $\prog{}$'s loop body.

The proof of Theorem~\ref{thm:sensitivity-computable} provides an alternative argument for admissible loops being sensitivity computable (Theorem~\ref{thm:admissible-diff});
as admissible loops do not contain defective variables by definition (Definition~\ref{def:admissible-loops}), the class of loops characterized by Theorem~\ref{thm:sensitivity-computable} subsumes the class of admissible loops.
\section{Experiments and Evaluation}\label{sec:evaluation}
We evaluate our methods for sensitivity analysis for admissible loops (Section~\ref{section:sensitvity:admissible}) and non-admissible loops (Section~\ref{section:sensitvity:non-admissible}).
Our techniques for sensitivity analysis extend the \Polar{} framework~\cite{polar}, which is publicly available at \url{https://github.com/probing-lab/polar}.
For admissible loops, we use the existing functionality of the \Polar{} framework to compute closed-forms for the moments of program variables. 

\paragraph{Experimental Setup.}
We split our evaluation into two parts. 
First, we compute the sensitivities of (higher) moments of program variables for admissible loops by automatically differentiating the closed-forms of the variables' moments (Table~\ref{tab:benchmark:admissible}).
In the second part, we consider our method using sensitivity recurrences, which is also applicable to non-admissible loops (Table~\ref{tab:benchmark:sensitivity-computable}).
To the best of our knowledge, our work provides the first exact and fully automatic tool to compute the sensitivities of (higher) moments of program variables for probabilistic loops.
All our experiments have been executed on a machine with a \SI{2.6}{GHz} Intel i7 (Gen 10) processor and \SI{32}{GB} of RAM with a timeout (TO) of \SI{120}{s}.

\paragraph{Differentiating Closed-Forms.}
Table~\ref{tab:benchmark:admissible} shows the evaluation of our sensitivity analysis technique for admissible loops (Section~\ref{section:sensitivity}) on $11$ benchmarks.
The benchmarks consist of the running example from Figure~\ref{fig:running:1} and parameterized probabilistic loops from the benchmarks in \cite{polar}, coming from literature on probabilistic program analysis \cite{Barthe2016,Chakarov2014,Gretz2013,BartocciKS20A}.
All the benchmarks contain at least one symbolic parameter with respect to which the sensitivities are computed.
Table~\ref{tab:benchmark:admissible} shows that our approach is capable of computing the sensitivities of higher moments of program variables for challenging loops with various characteristics, such as discrete and continuous state spaces as well as drawing from common distributions.
\begin{table}[htb]
  \setlength{\tabcolsep}{0.5em}
  \footnotesize
  \centering
  \begin{tabular}{@{}lllll@{}}
    \toprule
    \textsc{Benchmark} & \textsc{Sensitivity} & \textsc{Rec}, \textsc{RT} & \textsc{Sensitivity} & \textsc{Rec}, \textsc{RT}  \\
    \toprule
    50-Coin-Flips &
    $\partial_p \E(\text{total})$ & 51, 1.56 &
    $\partial_p \E(\text{total}^2)$ & TO,  TO \\ \midrule
    
    Bimodal &
    $\partial_{\text{p}} \E(\text{x})$ & 3, 0.40 &
    $\partial_{\text{p}} \E(\text{x}^2)$ & 5, 0.72 \\ \midrule
    
    Component-Health &
    $\partial_{\text{p1}} \E(\text{obs})$, &  2, 0.61 &
    $\partial_{\text{p1}} \E(\text{obs}^2)$, & 2, 0.62 \\ \midrule
    
    Umbrella &
    $\partial_{\text{u1}} \E(\text{umbrella})$ & 2, 0.97 &
    $\partial_{\text{u1}} \E(\text{umbrella}^2)$ & 2, 0.98 \\ \midrule
    
    Gambler's Ruin  &
    $\partial_{\text{p}} \E(\text{money})$ & 4, 11.2 &
    $\partial_{\text{p}} \E(\text{money}^2)$ & 10, 64.6 \\ \midrule
    
    Hawk-Dove &
    $\partial_{\text{v}} \E(\text{p1bal})$ & 1, 0.34 &
    $\partial_{\text{v}} \E(\text{p1bal}^2)$ & 2, 0.67 \\ \midrule
    
    Las-Vegas-Search & 
    $\partial_{\text{p}} \E(\text{attempts})$ & 2, 0.57 &
    $\partial_{\text{p}} \E(\text{attempts}^2)$ & 4, 7.31 \\ \midrule
    
    1D-Random-Walk  &
    $\partial_{\text{p}} \E(\text{x})$ & 1, 0.27 &
    $\partial_{\text{p}} \E(\text{x}^2)$ & 2, 0.39 \\ \midrule
    
    2D-Random-Walk &
    $\partial_{\text{p\_right}} \E(\text{x})$ & 1, 0.28 &
    $\partial_{\text{p\_right}} \E(\text{x}^2)$ & 2, 0.41 \\ \midrule
    
    Randomized-Response &
    $\partial_{\text{p}} \E(\text{p1})$ & 1, 0.29 &
    $\partial_{\text{p}} \E(\text{p1}^2)$ & 2, 0.42 \\ \midrule
    
    Vaccination (Fig.~\ref{fig:running:1}) &
    $\partial_{\text{vp}} \E(\text{infected})$ & 2, 1.25 &
    $\partial_{\text{vp}} \E(\text{infected}^2)$ & 2, 1.19 \\ \bottomrule
  \end{tabular}
  \vspace{0.5em}
  \caption{Evaluation of the sensitivity computation for $11$ admissible loops by differentiating closed-forms of variable moments. \textsc{Rec}: size of the recurrence system  to compute the variables' moments; \textsc{RT}: runtime in seconds; TO: timeout.}
  \label{tab:benchmark:admissible}
  \vspace{-1.5em}
\end{table}

\paragraph{Sensitivity Recurrences.}
Table~\ref{tab:benchmark:sensitivity-computable} shows the evaluation of our sensitivity analysis technique from Algorithm~\ref{alg:sensitivity-recurrences} using \emph{sensitivity recurrences}.
The benchmarks consist of four non-admissible loops and six admissible loops from Table~\ref{tab:benchmark:admissible}.
Non-admissible loops are known to be notoriously hard to analyze automatically~\cite{unsolvable}.
Table~\ref{tab:benchmark:sensitivity-computable} shows that \emph{sensitivity recurrences} are capable of computing 
the sensitivities for admissible as well as non-admissible loops.
\begin{table}[htb]
  \setlength{\tabcolsep}{0.5em}
  \footnotesize
  \centering
  \begin{tabular}{@{}lllll@{}}
    \toprule
    \textsc{Benchmark} & \textsc{Sensitivity} & \textsc{Rec}, \textsc{RT} & \textsc{Sensitivity} & \textsc{Rec}, \textsc{RT} \\
    \midrule

    Non-Admissible (Fig.~\ref{fig:running:2}) &
    $\partial_{\text{p}} \E(\text{u})$ & 9, 1.40 &
    $\partial_{\text{p}} \E(\text{y}^2)$ &  9, 1.75 \\ \midrule
    
    Non-Admissible-2 &
    $\partial_{\text{par}} \E(\text{y})$ & 5, 6.56 &
    $\partial_{\text{par}} \E(\text{xz})$ & 4, 3.67 \\ \midrule
    
    Non-Admissible-3 &
    $\partial_{\text{p}} \E(\text{total})$ & 6, 12.6 &
    $\partial_{\text{p}} \E(\text{z1}^2)$ & 12, 56.5 \\ \midrule
    
    Non-Admissible-4 &
    $\partial_{\text{p1}} \E(\text{z})$ & 4, 0.48 &
    $\partial_{\text{p1}} \E(\text{cnt}^2)$ & 3, 0.39 \\ \midrule
    
    Bimodal &
    $\partial_{\text{var}} \E(\text{x})$ & 3, 0.28 &
    $\partial_{\text{var}} \E(\text{x}^2)$ & 5, 0.42 \\ \midrule
    
    Component-Health &
    $\partial_{\text{p1}} \E(\text{obs})$ & 3, 0.74 &
    $\partial_{\text{p1}} \E(\text{obs}^2)$ & 3, 0.73\\ \midrule
    
    Gambler's Ruin &
    $\partial_{\text{p}} \E(\text{money})$ & 7, 66.9 &
    $\partial_{\text{p}} \E(\text{money}^2)$ & TO, TO\\ \midrule
    
    Las-Vegas-Search &
    $\partial_{\text{p}} \E(\text{attempts})$ & 3, 0.81 &
    $\partial_{\text{p}} \E(\text{attempts}^2)$ & 7, 13.3 \\ \midrule
    
    Randomized-Response &
    $\partial_{\text{p}} \E(\text{p1})$ & 1, 0.30 &
    $\partial_{\text{p}} \E(\text{p1}^2)$ & 3, 0.40\\ \midrule
    
    Vaccination (Fig.~\ref{fig:running:1}) &
    $\partial_{\text{vp}} \E(\text{infected})$ & 3, 8.26 &
    $\partial_{\text{vp}} \E(\text{infected}^2)$ & 3, 7.85 \\
    \bottomrule
  \end{tabular}
  \vspace{0.5em}
  \caption{Evaluation of the sensitivity computation for $10$ loops ($4$ are non-admissible) using sensitivity recurrences. \textsc{Rec}: size of the recurrence system to compute the variables' sensitivities; \textsc{RT}: runtime in seconds; TO: timeout.}
  \label{tab:benchmark:sensitivity-computable}
  \vspace{-1.5em}
\end{table}
\paragraph{Experimental Summary.}
When comparing both approaches on admissible loops, the differentiation-based approach typically performs better, e.g., on the benchmarks ``Gambler's Ruin'' or ``Vaccination''.
This is not surprising, as the main complexity in both approaches lies in solving the system of recurrences and when using sensitivity recurrences, the number of recurrences tends to be higher.
However, the exact number of recurrences depends on the program structure, and as such, there are cases where the approach using sensitivity recurrences performs equally well, such as in the ``Randomized-Response'' benchmark.
Nevertheless, for the class of loops characterized in Section~\ref{section:sensitvity:non-admissible}, the differentiation-based approach fails, whereas sensitivity recurrences still deliver exact results in a fully automated manner.

Our experiments demonstrate that our novel techniques for sensitivity analysis can compute the sensitivities for a rich class of probabilistic loops with discrete and continuous state spaces, drawing from probability distributions, and including polynomial arithmetic.
Moreover, the technique based on our new notion of \emph{sensitivity recurrences} can compute sensitivities for probabilistic loops for which closed-forms of the variables' moments do not exist.

\section{Related Work}

\paragraph{Sensitivity \& Probabilistic Programs.}
Bayesian networks can be seen as special loop-free probabilistic programs.
The sensitivity of Bayesian networks with discrete probability distribution was studied in \cite{ChanD02,ChanD04}.
The works of \cite{BartocciKS20A,StankovicBK22} provide a framework to analyze properties (sensitivity among others) of \emph{Prob-solvable Bayesian networks}.
In contrast, our work focuses on probabilistic loops with more complex control flow and supports continuous distributions.
In recent years, techniques emerged to manually reason about sensitivities of probabilistic programs, such as program calculi~\cite{Aguirre2021}, custom logics~\cite{BartheEGHS18}, or type systems~\cite{VasilenkoVB22}.
Although applicable to general probabilistic programs, these techniques require manual reasoning or user guidance, while our work focuses on full automation.

A fully-automatic and exact sensitivity analyzer for probabilistic programs with statically bounded loops was proposed in \cite{HuangWM18}.
In comparison, our work focuses on potentially unbounded loops.
The authors of \cite{WangFCDX20} introduce an automatable approach for expected sensitivity based on martingales.
Their technique proves that a given program is Lipschitz-continuous for \emph{some} Lipschitz constant.
In contrast, our work produces \emph{exact} sensitivities for unbounded loops and we characterize a class of loops for which our technique is complete.

\paragraph{Recurrences in Program Analysis.}
Recurrence equations are a common tool in program analysis.
The work of \cite{RCarbonellK04,RCarbonellK07} first introduced the idea of using linear recurrences and Gröbner basis computation to synthesize loop invariants.
This line of work has been further generalized in \cite{Kovacs08,Humenberger18} to support more general recurrences.
In \cite{FarzanK15,Kincaid18} the authors apply linear recurrences to more complex programs and combine it with over-approximation techniques.
The work \cite{BreckCKR20} combines recurrence techniques with template-based methods to analyze recursive procedures.
Recurrence equations were first used for the analysis of probabilistic loops in \cite{BartocciKS19} to synthesize so-called \emph{moment-based invariants}.
This approach was further generalized by \cite{polar}.
Our technique of \emph{sensitivity recurrences} is applicable to loops whose variables' moments do not satisfy linear recurrences. The recent work \cite{unsolvable} studies the synthesis of invariants for such loops, but does not address sensitivity analysis.

\section{Conclusion}\label{sec:conclusion}

We establish a fully automatic and exact technique to compute the sensitivities of higher moments of program variables for probabilistic loops.
Our method is applicable to probabilistic loops with potentially uncountable state spaces, complex control flow, polynomial assignments, and drawing from common probability distributions.
For admissible loops, we utilize closed-forms of the variables' moments obtained through linear recurrences.
Moreover, we propose the notion of \emph{sensitivity recurrences} enabling the sensitivity analysis for probabilistic loops whose moments do not admit closed-forms.
We characterize a class of loops for which we prove \emph{sensitivity recurrences} to be sound and complete.
Our experiments demonstrate the feasibility of our techniques on challenging benchmarks.

\subsubsection*{Acknowledgements}
This research was supported by the Vienna Science and Technology Fund WWTF 10.47379/ICT19018 grant ProbInG, the ERC Consolidator Grant ARTIST 101002685, the Austrian FWF SFB project SpyCoDe F8504, and the SecInt Doctoral College funded by TU Wien.

\bibliographystyle{splncs04}
\bibliography{references}

\newpage
\appendix
\section{Appendix: Probabilistic Loop Syntax}\label{appendix:syntax}

\begin{figure}
    {
    \footnotesize
    $lop \in \{ and, or \}$,
    $cop \in \{ =, \neq, <, >, \geq, \leq \}$,
    $Dist \in \{ \text{Bernoulli}, \text{Normal}, \text{Uniform}, \dots \}$
    \begin{grammar}
    	<sym> ::= "a" | "b" | $\dots$ <var> ::= "x" | "y" | $\dots$
    	
    	<const> ::= $r \in \R$ | <sym> | <const> ( "+" | "*" | "/" ) <const>
    	
    	<poly> ::= <const> | <var> | <poly> ("+" | "-" | "*") <poly> | <poly>"**n"
    	
    	<assign> ::= <var> "=" <assign\_right> | <var> "," <assign> "," <assign\_right>
    	
    	<categorical> ::= <poly> ("\{"<const>"\}" <poly>)* ["\{"<const>"\}"]
    	
    	<assign\_right> ::= <categorical> | Dist"("<const>$^*$")"
    	
    	<bexpr> ::= "true" ($\star$) | "false" | <poly> <cop> <poly> | "not" <bexpr> | <bexpr> <lop> <bexpr>
    	
    	<ifstmt> ::= "if" <bexpr>":" <statems> ("else if" <bexpr>":" <statems>)$^*$ ["else:" <statems>] "end"
    	
    	<statem> ::= <assign> | <ifstmt> \quad \quad \quad <statems> ::= <statem>$^+$
    	
    	<loop> ::= <statem>* "while" <bexpr> ":" <statems> "end"
    \end{grammar}
    }
    \caption{Grammar describing the syntax of probabilistic loops $\langle \textit{loop} \rangle$. \cite{polar}}
    \label{fig:syntax}
\end{figure}

\section{Appendix: Variable Dependency}\label{appendix:var-dep}
When we refer to $\depends{N}{}$ as the transitive closure of $\ddepends{}{}$ where at least one of the direct dependencies is non-linear we mean the following.
We say that $x \depends{N}{} y$ if 
\begin{itemize}
    \item either $x \ddepends{N}{} y$, or
    \item there exists some $z$ such that either
    \begin{itemize}
        \item $x \ddepends{}{} z$ and $z \depends{N}{} y$, or
        \item $x \ddepends{N}{} z$ and $z \depends{}{} y$
    \end{itemize}
\end{itemize}

\begin{definition}[Variable Dependency]
Let \prog{} be a probabilistic loop and $x, y \in \vars{}$.
We say that $x$ \emph{depends directly on} $y$, in symbols $x \ddepends{}{} y$, if $y$ appears in an assignment of $x$ or an assignment of $x$ occurs in an if-statement where $y$ appears in the if-condition.
Furthermore, we say that the dependency is \emph{non-linear}, in symbols $x \ddepends{N}{} y$, if $y$ appears non-linearly in an assignment of $x$.

\medskip

\noindent
Moreover, the transitive closure $\depends{}{}$ of $\ddepends{}{}$ is the smallest relation satisfying:
\begin{itemize}
    \item If $x \ddepends{}{} y$, then $x \depends{}{} y$, and
    \item if $x \depends{}{} y \land y \depends{}{} z$, then $x \depends{}{} z$.
\end{itemize}

\medskip

\noindent
The relation $\depends{N}{}$ formalizing transitive non-linear dependencies is the smallest relation satisfying:
\begin{itemize}
    \item If $x \ddepends{N}{} y$, then $x \depends{N}{} y$,
    \item if $x \depends{N}{} y \land y \depends{}{} z$, then $x \depends{N}{} z$, and
    \item if $x \depends{}{} y \land y \depends{N}{} z$, then $x \depends{N}{} z$.
\end{itemize}
\end{definition}

\begin{definition}[$p$-Influenced Dependency]
Let \prog{} be a program with parameter $p \in \params$ and $x,y \in \vars$ with $x \ddepends{}{} y$.
Then, the direct dependency between $x$ and $y$ is \emph{$p$-influenced}, in symbols $x \ddepends{}{p} y$ if at least one of the following conditions hold:

\begin{itemize}
    \item An assignment of $x$ contains $y$ and occurs in an if-statement with the if-condition involving $p$ or a $p$-dependent variable.
    \item An assignment of $x$ contains $y$ and is a probabilistic choice with some probability of the choice depending on $p$.
    \item An assignment of $x$ contains a term $c \cdot M \cdot y$ where $c$ is constant and $M$ is a monomial in program variables (possibly containing $y$). Moreover, either $c$ contains $p$ or $M$ contains a $p$-dependent variable.
\end{itemize}

If $x \depends{}{} y$, we write $x \depends{}{p} y$ if at least one of the dependencies from $x$ to $y$ is $p$-influenced. That means $x \depends{}{p} y$ if and only if
\begin{itemize}
    \item $x \ddepends{}{p} y$, or
    \item $x \ddepends{}{p} v \depends{}{} y$ for some variable $v$, or
    \item $x \depends{}{} v \ddepends{}{p} y$ for some variable $v$, or
    \item $x \depends{}{} v_1 \ddepends{}{p} v_2 \depends{}{} y$ for some variables $v_1, v_2$.
\end{itemize}

If $x \depends{}{} y$ and $x \not\depends{}{p} y$ we call the dependency between $x$ and $y$ \emph{$p$-free}.

\end{definition}

\section{Appendix: Moment Recurrences}\label{appendix:moment-recurrences}

For programs of our programming model (cf. Figure~\ref{appendix:syntax}) and a monomial $M$ in program variables, a moment recurrence for $M$ equates the expected value of $M$ at iteration $n{+}1$, that is $\E(M_{n+1})$ to expected values of program variable monomials at iteration $n$ (Definition~\ref{def:moment-rec}).
It is always possible to construct a moment recurrence for $M$ if all variables in all branching conditions are finite valued~\cite{polar}.
For completeness, we restate the process introduced in \cite{polar} on how moment recurrences are constructed.

\paragraph{Normalized Programs.}
To simplify the construction of moment recurrences, we can restrict ourselves to probabilistic loops satisfying the following conditions:
\begin{itemize}
    \item All distribution parameters are constant.
    \item The loop body is a flat sequence of guarded assignments.
    That means, every assignment is of the form $x = e_1 [C] e_2$, where $x$ is a program variable and $C$ is a boolean condition over program variables.
    The expression $e_1$ is either a distribution or a probabilistic choice of polynomials and is assigned to the variable $x$ if $C$ evaluates to \emph{true}.
    The expression $e_2$ is a single variable and is assigned to $x$ if $C$ evaluates to \emph{false}.
    \item In the loop body, every program variable is only assigned once.
\end{itemize}

Programs satisfying these conditions are called \emph{normalized}.
Every program from our program model can be transformed into a normalized program while maintaining the joint distribution of program variables as well as the dependencies between program variables as shown in \cite{polar}. 

\paragraph{Construction of Moment Recurrences.}
Let \prog{} be a normalized program with all program variables in branching conditions (or guards) being finitely-valued.
Given a monomial in program variables $M$, the moment recurrence for $M$ is constructed by starting with the expression $\E(M_{n+1})$ and replacing variables contained in the expression by their assignments bottom-up as they appear in the loop-body of \prog{}.
Throughout, the linearity of expectation is used to convert expected values of polynomials into expected values of monomials.
Because \prog{} is a normalized program, its loop body is a sequence of guarded assignments.
Assume $x$ is a program variable appearing in $\E(M_{n+1})$.
Therefore, $M = M' \cdot x_{n+1}^k$ for some monomial $M'$ free of $x$.
The assignment of $x$ is either a probabilistic choice of polynomials, 
$$x = a_0 \{p_0\} \dots \{p_{i-1}\} a_i\ [C]\ d,$$

for polynomials $a_0, \dots, a_i$ and constant probabilities $p_0, \dots, p_{i-1}$, or the assignment of $x$ is drawing from a known distribution, 
$$x = Dist\ [C]\ d.$$

Because \prog{} is a normalized program, $C$ is a boolean condition and $d$ is a program variable.
In case the assignment of $x$ is a (guarded) probabilistic choice of polynomials, $\E(M_{n+1})$ is rewritten to
\begin{equation*}
    \E(M_{n+1}) = \E(M'\cdot x^k_{n+1}) = \E\left(M'\left( d [\lnot C] + \sum p_i a_i^k [C]\right)\right).
\end{equation*}

The second option is that the assignment of $x$ is a (guarded) draw from a distribution. In this case $\E(M_{n+1})$ is rewritten to
\begin{equation*}
    \E(M_{n+1}) = \E(M'\cdot x^k_{n+1}) = \E\left(M' d [\lnot C]\right) + \E\left(M'[C]\right)\E\left(Dist^k\right).
\end{equation*}

In the above expressions $[C]$ denotes the Iverson bracket, evaluating to $1$ if $C$ holds and to $0$ otherwise.
By assumption, all program variables in branching conditions, and hence all variables in $C$, are finitely-valued.
Therefore, we can replace all occurrences of $[C]$ and $[\lnot C]$ in the above equations with polynomials over variables occurring in the condition $C$ as described in \cite{polar}.
Throughout the process, the linearity of expectation is used to turn expected values of polynomials into linear combinations of expected values of monomials.
By applying the replacement of program variables by their assignments for every variable bottom-up as they appear in the loop body, we end up with a the moment recurrence for $\E(M_{n+1})$.

\section{Appendix: Proof of Theorem~\ref{thm:sensitivity-computable}}\label{appendix:main-proof}

\sensitivitycomputable*

\begin{proof}
First, we cover the case where the monomial $M$ contains a defective variable.
If $M$ contains a defective variable $y$, then $y$ must be $p$-independent by condition~\ref{thm:sensitivity-computable:independent}.
As $M$ is a descendant of $x$, we have $x \depends{}{} y$.
Moreover, if there exists a $p$-dependent variable $z$ in $M$ different from $y$, Lemma~\ref{lemma:p-inf-moment-rec} gives us $x \depends{}{p} y$.
However, $x \depends{}{p} y$ contradicts our condition~\ref{thm:sensitivity-computable:p-free} that all dependencies on defective variables must be $p$-free.
Hence, the monomial $M$ is $p$-independent and Algorithm~\ref{alg:sensitivity-recurrences} terminates on line~\ref{algo:step:return-independent}.

\medskip

For the second, more involved case, assume all variables in $M$ are effective and possibly $p$-dependent.
Algorithm~\ref{alg:sensitivity-recurrences} does not terminate if and only if the algorithm adds infinitely many monomials $W'$ to the set $\sens$ one line~\ref{algo:step:add-sens} or to the set $\mom$ on lines \ref{algo:step:add-mom1} or \ref{algo:step:add-mom2}.
Every monomial $W'$ added to the set $\mom$ occurs in the moment recurrence of $W$ from line~\ref{algo:step:pick2}.
Moreover, every monomial $W'$ added to the set $\sens$ occurs in the sensitivity recurrence of $W$ from line~\ref{algo:step:pick1} and hence also occurs in the moment recurrence of $W$.
That is because the sensitivity recurrence and the moment recurrence of $W$ share the same monomials (Definition~\ref{def:sensitivity-recurrences}).

In \cite{polar}, the authors showed that every monomial $W'$ occurring in the moment recurrence of a monomial $W$ decreases with respect to a well-founded ordering if (A) all variables in branching conditions are finite, and (B) all variables in $W$ and $W'$ are effective.
Premise (A) matches our condition~\ref{thm:sensitivity-computable:finite}.
Therefore, to show that only finitely many monomials are added to $\sens$ and $\mom$ (and hence establish termination of Algorithm~\ref{alg:sensitivity-recurrences}), \emph{it suffices to show that all monomials $W'$ added to $\sens$ and $\mom$ only contain effective variables}.

First, note that every monomial $W'$ added to $\sens$ or $\mom$ is a descendant of the variable $x$.
This holds because the algorithm starts with $\sens=\{M\}, \mom = \emptyset$, the monomial $M$ is a descendant of $x$, and $W'$ occurs in the moment recurrence of some $W \in \sens \cup \mom$.

\medskip

\emph{Claim: All monomials $W'$ added to $\sens$ on line~\ref{algo:step:add-sens} only contain effective variables.}
Towards a contradiction, assume some monomial $W'$ is added to $\sens$ on line~\ref{algo:step:add-sens} and $W'$ contains a defective variable $y$.
By condition~\ref{thm:sensitivity-computable:independent}, $y$ is $p$-independent.
By Lemma~\ref{lemma:p-inf-moment-rec} and condition~\ref{thm:sensitivity-computable:p-free}, all variables in $W'$ are $p$-independent.
Hence, the monomial $W'$ is $p$-independent and $\diffp \E(W'_n)$ was replaced by $0$ on line~\ref{algo:step:var-independence}.
Therefore, $W'$ could not have been added to $\sens$ on line~\ref{algo:step:add-sens}.

\medskip

\emph{Claim: All monomials $W'$ added to $\mom$ on line~\ref{algo:step:add-mom1} only contain effective variables.}
Towards a contradiction, assume some monomial $W'$ is added to $\mom$ on line~\ref{algo:step:add-mom1} and $W'$ contains a defective variable $y$.
The monomial $W'$ occurs in the sensitivity recurrence of $W$ (fixed at line~\ref{algo:step:pick1}) with coefficient $(\nicefrac{\partial}{\partial p}c)$.
Therefore, $W'$ occurs in the moment recurrence of $W$ with coefficient $c$.
By Lemma~\ref{lemma:p-inf-moment-rec} and condition~\ref{thm:sensitivity-computable:p-free}, the constant $c$ does not contain the parameter $p$.
Hence, $(\nicefrac{\partial}{\partial p}c) = 0$ and $\E(W'_n)$ was replaced by $0$ on line~\ref{algo:step:const-independence}.
Therefore, $W'$ could not have been added to $\mom$ on line~\ref{algo:step:add-mom1}.

\medskip

\emph{Claim: All monomials $W'$ added to $\mom$ on line~\ref{algo:step:add-mom2} only contain effective variables.}
First, note that for all monomials $W'$ added to $\mom$ on line~\ref{algo:step:add-mom1}, the corresponding coefficient  $(\nicefrac{\partial}{\partial p}c) \neq 0$ and hence $c$ must contain the parameter $p$.
Therefore, by Lemma~\ref{lemma:p-inf-moment-rec}, for all variables $y$ in all monomials $W'$ added to $\mom$ in the first while-loop, we have $x \depends{}{p} y$.
By transitivity of $\depends{}{p}$, we get $x \depends{}{p} y$ for all variables $y$ in all monomials $W'$ added to $\mom$ on line~\ref{algo:step:add-mom2}.
Therefore, all $W'$ added to $\mom$ on line~\ref{algo:step:add-mom2} cannot contain defective variables by condition~\ref{thm:sensitivity-computable:p-free}.\qed
\end{proof}

\newpage
\section{Appendix: Non-Admissible Benchmarks}\label{appendix:non-admissible-benchmarks}

\begin{figure}
\begin{subfigure}[b]{0.4\linewidth}
\noindent
\begin{minipage}{\textwidth}
\begin{lstlisting}[style=appendix-code-style]
u,w,x,y,z = 0,1,2,3,4
while $\star$:
  z = z+p$^2$ {1/2} z+p
  y = y - 5*p*z
  w = 5*w + x$^2$
  x = 5 + w+x
  u = x + p*z*y
end
\end{lstlisting}
\end{minipage}
\caption{Non-Admissible (Fig.~\ref{fig:running:2})}
\end{subfigure} \hfill
\begin{subfigure}[b]{0.52\linewidth}
\noindent
\begin{minipage}{\textwidth}
\begin{lstlisting}[style=appendix-code-style]
x,y,z,var = 1,2,a,0
d1,d2 = 5,3
run = -1
while $\star$:
  run = 2*run + z$^2$
  z = z+1
  d1,d2 = d1*d2+3, d1+z
  x = 3*x + d2 + par$^2$*z + run*z
  y = 3*(x-y) + par$^2$*run
end
\end{lstlisting}
\end{minipage}
\caption{Non-Admissible-2}
\end{subfigure}
\begin{subfigure}[b]{0.52\linewidth}
\noindent
\begin{minipage}{\textwidth}
\begin{lstlisting}[style=appendix-code-style]
cnt,total = 0,0
x1,x2 = 1,2
y1,y2 = 0,3
z1,z2 = 1,5
while $\star$:
  cnt = cnt + 1
  x1 = x1$^2$ + q*x2
  x2 = y1 + cnt + q
  y1 = r*(y1-cnt) + y2*cnt
  y2 = r*y1 + 5
  z1 = cnt$^2$ - cnt + p*z1
  z2 = z1*3 - 5*(z2-p)
  total = x2 + y2 + z2
end
\end{lstlisting}
\end{minipage}
\caption{Non-Admissible-3}
\end{subfigure}\hfill
\begin{subfigure}[b]{0.4\linewidth}
\noindent
\begin{minipage}{\textwidth}
\begin{lstlisting}[style=appendix-code-style]
y,x,z,cnt = 0,0,0,0
while $\star$:
  x = $\text{DiscreteUniform}$(1,5)
  if x < 3:
    inc = $\text{Bernoulli}$(p1)
    cnt = cnt + inc
  else:
    inc = $\text{Bernoulli}$(p2)
    cnt = cnt - inc
  end
  f = $\text{DiscreteUniform}$(0,10)
  y = y$^2$ + x * f
  z = cnt$^2$ - 3*y$^2$ + x$^3$
end
\end{lstlisting}
\end{minipage}
\caption{Non-Admissible-4}
\end{subfigure}
\caption{Four parameterized non-admissible loops used for our experiments (Section~\ref{sec:evaluation}).}
\end{figure}

\end{document}